\documentclass[11pt]{article}
\usepackage{hyperref}
\usepackage{times}  
\usepackage{mathpazo}
\usepackage{amssymb,amsmath,amsthm}
\usepackage{epsfig}
\usepackage{enumerate}
\usepackage{algorithm}
\usepackage{algpseudocode}

\newtheorem{theorem}{Theorem}[section]
\newtheorem{proposition}[theorem]{Proposition}
\newtheorem{definition}[theorem]{Definition}

\newtheorem{lemma}[theorem]{Lemma}

\newtheorem{corollary}[theorem]{Corollary}
\newtheorem{fact}[theorem]{Fact}

\newtheorem{remark}[theorem]{Remark}

\newcommand{\qedsymb}{\hfill{\rule{2mm}{2mm}}}
\renewenvironment{proof}[1][]{\begin{trivlist}
\item[\hspace{\labelsep}{\bf\noindent Proof#1:\/}] }{\qedsymb\end{trivlist}}

\def\calF{{\cal F}}
\def\calE{{\cal E}}

\def\calM{{\cal M}}

\def\Q{\mathbb{Q}}

\newcommand\Prob[2]{{\Pr_{#1}\left[ {#2} \right]}}

\newcommand{\NP}{\mathsf{NP}}

\renewcommand{\epsilon}{\varepsilon}

\newcommand{\LEAF}{\textsc{Leaf}}
\newcommand{\SchrijverP}{\textsc{Schrijver}}
\newcommand{\KneserP}{\textsc{Kneser}}
\newcommand{\Agree}{\textsc{Agreeable-Set}}

\newcommand{\PPA}{\mathsf{PPA}}
\newcommand{\TFNP}{\mathsf{TFNP}}

\begin{document}

\title{{\bf A Fixed-Parameter Algorithm for the Kneser Problem}}

\author{
Ishay Haviv\thanks{School of Computer Science, The Academic College of Tel Aviv-Yaffo, Tel Aviv 61083, Israel. Research supported in part by the Israel Science Foundation (grant No.~1218/20).
}
}

\date{}

\maketitle

\begin{abstract}
The Kneser graph $K(n,k)$ is defined for integers $n$ and $k$ with $n \geq 2k$ as the graph whose vertices are all the $k$-subsets of $\{1,2,\ldots,n\}$ where two such sets are adjacent if they are disjoint. A classical result of Lov\'asz asserts that the chromatic number of $K(n,k)$ is $n-2k+2$. In the computational $\KneserP$ problem, we are given an oracle access to a coloring of the vertices of $K(n,k)$ with $n-2k+1$ colors, and the goal is to find a monochromatic edge. We present a randomized algorithm for the $\KneserP$ problem with running time $n^{O(1)} \cdot k^{O(k)}$. This shows that the problem is fixed-parameter tractable with respect to the parameter $k$. The analysis involves structural results on intersecting families and on induced subgraphs of Kneser graphs.

We also study the $\Agree$ problem of assigning a small subset of a set of $m$ items to a group of $\ell$ agents, so that all
agents value the subset at least as much as its complement.
As an application of our algorithm for the $\KneserP$ problem, we obtain a randomized polynomial-time algorithm for the $\Agree$ problem for instances that satisfy $\ell \geq m - O(\frac{\log m}{\log \log m})$.
We further show that the $\Agree$ problem is at least as hard as a variant of the $\KneserP$ problem with an extended access to the input coloring.
\end{abstract}


\section{Introduction}

The Kneser graph $K(n,k)$ is defined for integers $n$ and $k$ with $n \geq 2k$ as the graph whose vertices are all the $k$-subsets of $[n] = \{1,2,\ldots,n\}$ where two such sets are adjacent if they are disjoint. In 1955, Kneser~\cite{Kneser55} observed that there exists a proper coloring of the vertices of $K(n,k)$ with $n-2k+2$ colors and conjectured that fewer colors do not suffice, that is, that its chromatic number satisfies $\chi(K(n,k)) = n-2k+2$.
The conjecture was proved more than two decades later by Lov\'asz~\cite{LovaszKneser} as a surprising application of the Borsuk-Ulam theorem from algebraic topology~\cite{Borsuk33}. Following this result, topological methods have become a common and powerful tool in combinatorics, discrete geometry, and theoretical computer science (see, e.g.,~\cite{MatousekBook}).
Several alternative proofs of Kneser's conjecture were provided in the literature over the years (see, e.g.,~\cite{MatousekZ04}), and despite the combinatorial nature of the conjecture, all of them essentially rely on the topological Borsuk-Ulam theorem. One exception is a proof by Matou{\v{s}}ek~\cite{Matousek04}, which is presented in a combinatorial form, but is yet inspired by a discrete variant of the Borsuk-Ulam theorem known as Tucker's lemma.

In the computational $\KneserP$ problem, we are given an access to a coloring of the vertices of $K(n,k)$ with $n-2k+1$ colors, and the goal is to find a monochromatic edge, i.e., two vertices with the same color that correspond to disjoint sets.
Since the number of colors used by the input coloring is strictly smaller than the chromatic number of $K(n,k)$~\cite{LovaszKneser}, it follows that every instance of the problem has a solution.
However, the topological proofs of the lower bound on the chromatic number of $K(n,k)$ are not constructive, in the sense that they do not supply an efficient algorithm for finding a monochromatic edge.
By an efficient algorithm we mean that its running time is polynomial in $n$, whereas the number of vertices $\binom{n}{k}$ might be exponentially larger.
Hence, it is natural to assume that the input coloring is given as an access to an oracle that given a vertex of $K(n,k)$ returns its color.
Alternatively, the coloring can be given by some succinct representation, e.g., a Boolean circuit that computes the color of any given vertex.

The question of determining the complexity of the $\KneserP$ problem was proposed by Deng, Feng, and Kulkarni~\cite{DengFK17}, who asked whether it is complete in the complexity class $\PPA$.
This complexity class belongs to a family of classes that were introduced by Papadimitriou~\cite{Papa94} in the attempt to characterize the mathematical arguments that lie behind the existence of solutions to search problems of $\TFNP$.
The complexity class $\TFNP$, introduced in~\cite{MegiddoP91}, is the class of all total search problems in $\NP$, namely, the search problems in which a solution is guaranteed to exist and can be verified in polynomial running time.
Papadimitriou has introduced in~\cite{Papa94} several subclasses of $\TFNP$, each of which consists of the total search problems that can be reduced to a problem that represents some mathematical argument.
In particular, the class $\PPA$ (Polynomial Parity Argument) corresponds to the simple fact that every (undirected) graph with maximum degree $2$ that has a vertex of degree $1$ must have another degree $1$ vertex. Hence, $\PPA$ is defined as the class of all problems in $\TFNP$ that can be efficiently reduced to the $\LEAF$ problem, in which given a succinct representation of a graph with maximum degree $2$ and given a vertex of degree $1$ in the graph, the goal is to find another such vertex.

In recent years, it has been shown that the complexity class $\PPA$ perfectly captures the complexity of several total search problems for which the existence of the solution relies on the Borsuk-Ulam theorem. Indeed, Filos-Ratsikas and Goldberg~\cite{FG18,FG19} proved that the Consensus Halving problem with inverse-polynomial precision parameter is $\PPA$-complete and derived the $\PPA$-completeness of some classical problems, such as the Splitting Necklace problem with two thieves and the Discrete Sandwich problem.
The $\PPA$-hardness of the Consensus Halving problem was further extended to a constant precision parameter in a recent work of Deligkas, Fearnley, Hollender, and Melissourgos~\cite{DeligkasFHM22}.
Another $\PPA$-complete problem, studied in~\cite{Haviv21} and closely related to the $\KneserP$ problem, is the $\SchrijverP$ problem which given a coloring of the Schrijver graph $S(n,k)$ with fewer colors than its chromatic number asks to find a monochromatic edge. Note that $S(n,k)$ is the induced subgraph of $K(n,k)$ on the collection of all $k$-subsets of $[n]$ with no two consecutive elements modulo $n$, and that its chromatic number is equal to that of $K(n,k)$~\cite{SchrijverKneser78}.
Despite the progress in understanding the complexity of total search problems related to the Borsuk-Ulam theorem, the question of~\cite{DengFK17} on the complexity of the $\KneserP$ problem is still open. The question of determining the complexity of its extension to Kneser hypergraphs was recently raised by Filos-Ratsikas, Hollender, Sotiraki, and Zampetakis~\cite{Filos-RatsikasH21}.

The study of the $\KneserP$ problem has been recently motivated by its connection to a problem called $\Agree$ that was introduced by Manurangsi and Suksompong~\cite{ManurangsiS19} and further studied by Goldberg, Hollender, Igarashi, Manurangsi, and Suksompong~\cite{GoldbergHIMS20}.
This problem falls into the category of resource allocation problems, where one assigns items from a given collection $[m]$ to $\ell$ agents that have different preferences. The preferences are given by monotone utility functions that associate a non-negative value to each subset of $[m]$. In the $\Agree$ setting, the agents act as a group, and the goal is to collectively allocate a subset of items that is agreeable to all of them, in the sense that every agent likes it at least as much as it likes the complement set. The authors of~\cite{ManurangsiS19} proved that for every $\ell$ agents with monotone utility functions defined on the subsets of $[m]$, there exists a subset $S \subseteq [m]$ of size
\begin{eqnarray}\label{eq:agreeable_size}
|S| \leq \min \Big ( \Big \lfloor \frac{m+\ell}{2} \Big \rfloor,m \Big )
\end{eqnarray}
that is agreeable to all agents, and that this bound is tight in the worst case. They initiated the study of the $\Agree$ problem that given an oracle access to the utility functions of the $\ell$ agents, asks to find a subset $S \subseteq [m]$ that satisfies the worst-case bound given in~\eqref{eq:agreeable_size} and that is agreeable to all agents (see Theorem~\ref{thm:MS} and Definition~\ref{def:Agree}).
Note that for instances with $\ell \geq m$, the collection $[m]$ forms a proper solution.

Interestingly, the proof of~\cite{ManurangsiS19} for the existence of a solution to the $\Agree$ problem relies on the chromatic number of Kneser graphs.
In fact, the proof implicitly shows that the $\Agree$ problem is efficiently reducible to the $\KneserP$ problem, hence the $\KneserP$ problem is at least as hard as the $\Agree$ problem.
An alternative existence proof, based on the Consensus Halving theorem~\cite{SimmonsS03}, was given by the authors of~\cite{GoldbergHIMS20}. Their approach was applied there to show that the $\Agree$ problem can be solved in polynomial time on instances with additive utility functions and on instances in which the number $\ell$ of agents is a fixed constant. The complexity of the problem in the general case is still open.

\subsection{Our Contribution}

Our main result concerns the parameterized complexity of the $\KneserP$ problem.
We prove that the problem is fixed-parameter tractable with respect to the parameter $k$, namely, it can be solved on an input coloring of a Kneser graph $K(n,k)$ in running time $n^{O(1)} \cdot f(k)$ for some function $f$.

\begin{theorem}\label{thm:AlgoKneserNew}
There exists a randomized algorithm that given integers $n$ and $k$ with $n \geq 2k$ and an oracle access to a coloring $c: \binom{[n]}{k} \rightarrow [n-2k+1]$ of the vertices of the Kneser graph $K(n,k)$, runs in time $n^{O(1)} \cdot k^{O(k)}$ and returns a monochromatic edge with probability $1-2^{-\Omega(n)}$.
\end{theorem}
It should be mentioned that the notion of fixed-parameter tractability is used here in a slightly different manner than in the parameterized complexity literature (see, e.g.,~\cite{CyganFKLMPPS15}).
This is because Theorem~\ref{thm:AlgoKneserNew} deals with the parameterized complexity of a total search problem, rather than a decision problem, and because its input is given as an oracle access.
In fact, borrowing the terminology of this area, our algorithm for the $\KneserP$ problem can be viewed as a randomized polynomial Turing kernelization algorithm for the problem (see, e.g.,~\cite[Chapter~22]{KernelBook19}).
Namely, we show that the problem of finding a monochromatic edge in a Kneser graph $K(n,k)$ can essentially be reduced by a randomized efficient algorithm to finding a monochromatic edge in a Kneser graph $K(n',k)$ for $n' = O(k^4)$ (see Section~\ref{sec:kernel} for the precise details).

The analysis of the algorithm given in Theorem~\ref{thm:AlgoKneserNew} relies on properties of induced subgraphs of Kneser graphs (see Section~\ref{sec:properties}).
The proofs involve a stability result due to Hilton and Milner~\cite{HM67} of the celebrated Erd{\"{o}}s-Ko-Rado theorem~\cite{EKR61} (see Theorem~\ref{thm:HM}) and an idea recently applied by Frankl and Kupavskii~\cite{FrankK20} in the study of maximal degrees in induced subgraphs of Kneser graphs.
An overview of the proof of Theorem~\ref{thm:AlgoKneserNew} is given in Section~\ref{sec:overview}.

Our next result provides a simple deterministic algorithm for the $\KneserP$ problem that is particularly useful for Kneser graphs $K(n,k)$ with $k$ close to $n/2$.
Its analysis is based on the chromatic number of Schrijver graphs~\cite{SchrijverKneser78}.
\begin{theorem}\label{thm:KneserAlgSch}
There exists an algorithm that given integers $n$ and $k$ with $n \geq 2k$ and an oracle access to a coloring $c: \binom{[n]}{k} \rightarrow [n-2k+1]$ of the vertices of the Kneser graph $K(n,k)$, returns a monochromatic edge in running time polynomial in $\binom{n-k+1}{k} \leq n^{\min(k,n-2k+1)}$.
\end{theorem}

We proceed by presenting our results on the $\Agree$ problem.
First, as applications of Theorems~\ref{thm:AlgoKneserNew} and~\ref{thm:KneserAlgSch}, using the relation from~\cite{ManurangsiS19} between the $\Agree$ and $\KneserP$ problems, we obtain the following algorithmic results.

\begin{theorem}\label{thm:AlgoIntro1}
There exists a randomized algorithm for the $\Agree$ problem that given an oracle access to an instance with $m$ items and $\ell$ agents ($\ell<m$), runs in time $m^{O(1)} \cdot k^{O(k)}$ for $k = \lceil \frac{m-\ell}{2} \rceil$ and returns a proper solution with probability $1-2^{-\Omega(m)}$.
\end{theorem}

\begin{theorem}\label{thm:AlgoIntro2}
There exists an algorithm for the $\Agree$ problem that given an oracle access to an instance with $m$ items and $\ell$ agents ($\ell<m$), returns a proper solution in running time polynomial in $\binom{m-k+1}{k} \leq m^{\min(k,m-2k+1)}$ for $k = \lceil \frac{m-\ell}{2} \rceil$.
\end{theorem}
\noindent
We apply Theorems~\ref{thm:AlgoIntro1} and~\ref{thm:AlgoIntro2} to show that the $\Agree$ problem can be solved in polynomial time for certain families of instances.
By Theorem~\ref{thm:AlgoIntro1}, we obtain a randomized efficient algorithm for instances in which the number of agents $\ell$ is not much smaller than the number of items $m$, namely, for $\ell \geq m - O(\frac{\log m}{\log \log m})$ (see Corollary~\ref{cor:m<=l+log/loglog}).
By Theorem~\ref{thm:AlgoIntro2}, we obtain an efficient algorithm for instances with a constant number of agents (see Corollary~\ref{cor:constant_l}), providing an alternative proof for a result of~\cite{GoldbergHIMS20}.

We finally explore the relations between the $\Agree$ and $\KneserP$ problems.
As already mentioned, there exists an efficient reduction from the $\Agree$ problem to the $\KneserP$ problem~\cite{ManurangsiS19}.
Here we provide a reduction in the opposite direction.
However, for the reduction to be efficient we reduce from a variant of the $\KneserP$ problem with an extended type of queries, which we call {\em subset queries} and define as follows.
For an input coloring $c : \binom{[n]}{k} \rightarrow [n-2k+1]$ of a Kneser graph $K(n,k)$, a subset query is a pair $(i,B)$ of a color $i \in [n-2k+1]$ and a set $B \subseteq [n]$, and the answer on the query $(i,B)$ determines whether $B$ contains a vertex colored $i$, that is, whether there exists a $k$-subset $A \subseteq B$ satisfying $c(A)=i$.
We prove the following result (see Section~\ref{sec:models} for the computational input model of the problems).

\begin{theorem}\label{thm:KneserVSAgree}
There exists a polynomial-time reduction from the $\KneserP$ problem with subset queries to the $\Agree$ problem.
\end{theorem}

\subsection{Proof Overview of Theorem~\ref{thm:AlgoKneserNew}}\label{sec:overview}

We present here the main ideas of our fixed-parameter algorithm for the $\KneserP$ problem.
Suppose that we are given, for $n \geq 2k$, an oracle access to a coloring $c: \binom{[n]}{k} \rightarrow [n-2k+1]$ of the vertices of the Kneser graph $K(n,k)$ with $n-2k+1$ colors. As mentioned before, since the chromatic number of $K(n,k)$ is $n-2k+2$~\cite{LovaszKneser}, the coloring $c$ must have a monochromatic edge. Such an edge can clearly be found by an algorithm that queries the oracle for the colors of all the vertices.
However, the running time of such an algorithm is polynomial in $\binom{n}{k}$, so it is not fixed-parameter with respect to the parameter $k$.

A natural attempt to improve on this running time is to consider a randomized algorithm that picks uniformly and independently at random polynomially many vertices of $K(n,k)$ and checks if any two of them form a monochromatic edge. However, it is not difficult to see that there exist colorings of $K(n,k)$ with $n-2k+1$ colors for which a small fraction of the vertices are involved in all the monochromatic edges, implying that the success probability of this randomized algorithm is negligible on them.
To see this, consider the canonical coloring of $K(n,k)$, in which for every $i \in [n-2k+1]$ the vertices $A \in \binom{[n]}{k}$ with $\min(A)=i$ are colored $i$, and the remaining vertices, those contained in $[n] \setminus [n-2k+1]$, are colored $n-2k+2$. By recoloring the vertices of the last color class with arbitrary colors from $[n-2k+1]$, we get a coloring of $K(n,k)$ with $n-2k+1$ colors such that every monochromatic edge has an endpoint in a collection of $\binom{2k-1}{k}$ vertices.
This implies that the probability that two random vertices chosen uniformly and independently from $\binom{[n]}{k}$ form a monochromatic edge may go to zero faster than any inverse polynomial in $n$.

While the above coloring shows that all the monochromatic edges can involve vertices from a small set, one may notice that this coloring is very well structured, in the sense that each color class is quite close to a trivial intersecting family, i.e., an intersecting family all of whose members share a common element.
This is definitely not a coincidence. It is known that large intersecting families of $k$-subsets of $[n]$ are `essentially' contained in trivial intersecting families.
Indeed, the classical Erd{\"{o}}s-Ko-Rado theorem~\cite{EKR61} asserts that the largest size of an intersecting family of $k$-subsets of $[n]$ is $\binom{n-1}{k-1}$, attained by, and only by, the $n$ largest trivial families (for $n>2k$). Moreover, Hilton and Milner~\cite{HM67} proved a stability result for the Erd{\"{o}}s-Ko-Rado theorem, saying that if an intersecting family of sets from $\binom{[n]}{k}$ is not trivial then its size cannot exceed $\binom{n-1}{k-1}-\binom{n-k-1}{k-1}+1$, which is much smaller than the largest possible size of an intersecting family when $n$ is sufficiently larger than $k$ (see Theorem~\ref{thm:HM} and Remark~\ref{remark:HM}). More recently, it was shown by Dinur and Friedgut~\cite{DinurF09} that every intersecting family can be made trivial by removing not more than $\tilde{c} \cdot \binom{n-2}{k-2}$ of its members for a constant $\tilde{c}$ (see~\cite{Kupavskii18} for an exact $\binom{n-3}{k-2}$ bound on the number of sets that should be removed, provided that $n \geq \tilde{c} \cdot k$ for a constant $\tilde{c}$).
Hence, our strategy for finding a monochromatic edge in $K(n,k)$ is to learn the structure of the large color classes which are close to being intersecting.
We use random samples from the vertex set of the graph in order to identify the common elements of the trivial families that `essentially' contain these color classes.
Roughly speaking, this allows us to repeatedly reduce the size of the ground set $[n]$ of the given Kneser graph and to obtain a small subgraph, whose size depends only on $k$, that is expected to contain a monochromatic edge. Then, such an edge can simply be found by querying the oracle for the colors of all the vertices of this subgraph.

With the above idea in mind, let us consider an algorithm that starts by selecting uniformly and independently at random polynomially many vertices of $K(n,k)$ and queries the oracle for their colors. Among the $n-2k+1$ color classes, there must be a non-negligible one that includes at least $\frac{1}{n-2k+1} \geq \frac{1}{n}$ fraction of the vertices. It can be shown, using the Chernoff-Hoeffding bound, that the samples of the algorithm can be used to learn a color $i \in [n-2k+1]$ of a quite large color class. Moreover, the samples can be used to identify an element $j \in [n]$ that is particularly popular on the vertices colored $i$ (say, that belongs to a constant fraction of them) in case that such an element exists. Now, if the coloring satisfies that (a) there exists an element $j$ that is popular on the vertices colored $i$ and, moreover, (b) this element $j$ belongs to {\em all} the vertices that are colored $i$, then it suffices to focus on the subgraph of $K(n,k)$ induced by the vertices of $\binom{[n] \setminus \{j\}}{k}$. Indeed, condition (b) implies that the restriction of the given coloring to this subgraph uses at most $n-2k$ colors.
Since this subgraph is isomorphic to $K(n-1,k)$, its chromatic number is $n-2k+1$, hence it has a monochromatic edge. By repeatedly applying this procedure, assuming that the conditions (a) and (b) hold in all iterations, we can eliminate elements from the ground set $[n]$ and obtain smaller and smaller Kneser graphs that still have a monochromatic edge. When the ground set becomes sufficiently small, one can go over all the remaining vertices and efficiently find the required edge.
We now turn to address the case where at least one of the conditions (a) and (b) does not hold.

For condition (a), suppose that in some iteration the algorithm identifies a color $i$ that appears on a significant fraction of vertices, but no element of $[n]$ is popular on these vertices.
In this case, one might expect the color class of $i$ to be so far from being intersecting, so that the polynomially many samples would include with high probability a monochromatic edge of vertices colored $i$.
It can be observed that the aforementioned stability results of the Erd{\"{o}}s-Ko-Rado theorem yield that a non-negligible fraction of the vertices of a large color class with no popular element lie on monochromatic edges.
However, in order to catch a monochromatic edge with high probability we need here the stronger requirement, saying that a non-negligible fraction of the {\em pairs} of vertices from this color class form monochromatic edges. We prove that this indeed holds for color classes of $K(n,k)$ that include at least, say, $\frac{1}{n}$ fraction of the vertices, provided that $n \geq \tilde{c} \cdot k^4$ for some constant $\tilde{c}$. The proof uses the Hilton-Milner theorem and borrows an idea of~\cite{FrankK20} (see Lemma~\ref{lemma:at_most_gamma} and Corollary~\ref{cor:at_most_gamma}).

For condition (b), suppose that in some iteration the algorithm identifies a color $i$ of a large color class and an element $j$ that is popular in its sets but does not belong to all of them. In this case, we can no longer ensure that the final Kneser graph obtained after all iterations has a monochromatic edge, as some of its vertices might be colored $i$.
To handle this situation, we show that every vertex $A \in \binom{[n]}{k}$ with $j \notin A$ has a lot of neighbors colored $i$.
Hence, if the algorithm finds at its final step, or even earlier, a vertex $A$ colored $i$ satisfying $j \notin A$, where $i$ and $j$ are the color and element chosen by the algorithm in a previous iteration, then it goes back to the ground set of this iteration and finds using additional polynomially many random vertices from it a neighbor of $A$ colored $i$, and thus a monochromatic edge (see Lemma~\ref{lemma:at_least_gamma} and Corollary~\ref{cor:at_least_gamma}).

To summarize, our algorithm for the $\KneserP$ problem repeatedly calls an algorithm, which we refer to as the `element elimination' algorithm, that uses polynomially many random vertices to identify a color $i$ of a large color class. If no element of $[n]$ is popular on this color class then the random samples provide a monochromatic edge with high probability and we are done. Otherwise, the algorithm finds such a popular $j \in [n]$ and focuses on the subgraph obtained by eliminating $j$ from the ground set. This is done as long as the size of the ground set is larger than $\tilde{c} \cdot k^4$ for a constant $\tilde{c}$. After all iterations, the remaining vertices induce a Kneser graph $K(n',k)$ for $n' \leq \tilde{c} \cdot k^4$, and the algorithm queries for the colors of all of its vertices in time polynomial in $\binom{n'}{k} \leq k^{O(k)}$. If the colors that were chosen through the iterations of the `element elimination' algorithm do not appear in this subgraph, then it must contain a monochromatic edge which can be found by an exhaustive search. Otherwise, this search gives us a vertex $A$ satisfying $c(A) = i$ and $j \notin A$ for a color $i$ associated with an element $j$ by one of the calls to the `element elimination' algorithm. As explained above, given such an $A$ it is possible to efficiently find with high probability a vertex that forms with $A$ a monochromatic edge.
This gives us a randomized algorithm with running time $n^{O(1)} \cdot k^{O(k)}$ that finds a monochromatic edge with high probability.
The full description of the algorithm and its analysis are presented in Section~\ref{sec:algo}.

\subsection{Outline}
The rest of the paper is organized as follows.
In Section~\ref{sec:pre}, we gather several definitions and results that will be used throughout the paper.
In Section~\ref{sec:algo}, we present and analyze our randomized fixed-parameter algorithm for the $\KneserP$ problem and prove Theorem~\ref{thm:AlgoKneserNew}.
In Section~\ref{sec:AlgoSchr}, we present a simple deterministic algorithm for the $\KneserP$ problem and prove Theorem~\ref{thm:KneserAlgSch}.
In Section~\ref{sec:agree}, we study the $\Agree$ problem. We prove there Theorems~\ref{thm:AlgoIntro1} and~\ref{thm:AlgoIntro2} and derive efficient algorithms for certain families of instances.
We also show that the $\Agree$ problem is at least as hard as the $\KneserP$ problem with subset queries, confirming Theorem~\ref{thm:KneserVSAgree}.

\section{Preliminaries}\label{sec:pre}

\subsection{Computational Models}\label{sec:models}

In this work we consider total search problems whose inputs involve functions that are defined on domains of size exponential in the parameters of the problems.
For example, the input of the $\KneserP$ problem is a coloring $c: \binom{[n]}{k} \rightarrow [n-2k+1]$ of the vertices of the Kneser graph $K(n,k)$ for $n \geq 2k$.
For such problems, one has to specify how the input is given. We consider the following two input models.
\begin{itemize}
  \item In the {\em black-box input model}, an input function is given as an oracle access, so that an algorithm can query the oracle for the value of the function on any element of its domain. This input model is used in the current work to present our algorithmic results, reflecting the fact that the algorithms do not rely on the representation of the input functions.
  \item In the {\em white-box input model}, an input function is given by a succinct representation that can be used to efficiently determine the values of the function, e.g., a Boolean circuit or an efficient Turing machine. This input model is appropriate to study the computational complexity of problems, and in particular, to show membership and hardness results with respect to the complexity class $\PPA$.
\end{itemize}

Reductions form a useful tool to show relations between problems.
Let $P_1$ and $P_2$ be total search problems.
We say that $P_1$ is (polynomial-time) reducible to $P_2$ if there exist (polynomial-time) computable functions $f,g$ such that $f$ maps any input $x$ of $P_1$ to an input $f(x)$ of $P_2$, and $g$ maps any pair $(x,y)$ of an input $x$ of $P_1$ and a solution $y$ of $f(x)$ with respect to $P_2$ to a solution of $x$ with respect to $P_1$.
For problems $P_1$ and $P_2$ in the black-box input model, one has to use the notion of {\em black-box reductions}. A (polynomial-time) black-box reduction satisfies that the oracle access needed for the input $f(x)$ of $P_2$ can be simulated by a (polynomial-time) procedure that has an oracle access to the input $x$. In addition, the solution $g(x,y)$ of $x$ in $P_1$ can be computed (in polynomial time) given the solution $y$ of $f(x)$ and the oracle access to the input $x$.
In the current work we will use black-box reductions to obtain algorithmic results for problems in the black-box input model.
For more details on these concepts, we refer the reader to~\cite[Section~2.2]{BeameCEIP98} (see also~\cite[Sections~2 and~4]{DeligkasFH21} for related discussions).

\subsection{Kneser and Schrijver Graphs}\label{sec:KneSchr}

Consider the following definition.
\begin{definition}\label{def:K(F)}
For a family $\calF$ of non-empty sets, let $K(\calF)$ denote the graph on the vertex set $\calF$ in which two vertices are adjacent if they represent disjoint sets.
\end{definition}
\noindent
For a set $X$ and an integer $k$, let $\binom{X}{k}$ denote the family of all $k$-subsets of $X$.
Equipped with Definition~\ref{def:K(F)}, the {\em Kneser graph} $K(n,k)$ can be defined for integers $n$ and $k$ with $n \geq 2k$ as the graph $K(\binom{[n]}{k})$.
A set $A \subseteq [n]$ is said to be {\em stable} if it includes no two consecutive elements modulo $n$, that is, it forms an independent set in the cycle $C_n$ with the numbering from $1$ to $n$ along the cycle. For integers $n$ and $k$ with $n \geq 2k$, let $\binom{[n]}{k}_{\mathrm{stab}}$ denote the collection of all stable $k$-subsets of $[n]$. The {\em Schrijver graph} $S(n,k)$ is defined as the graph $K(\binom{[n]}{k}_{\mathrm{stab}})$. Equivalently, it is the induced subgraph of $K(n,k)$ on the vertex set $\binom{[n]}{k}_{\mathrm{stab}}$.

The chromatic numbers of the graphs $K(n,k)$ and $S(n,k)$ were determined, respectively, by Lov\'asz~\cite{LovaszKneser} and by Schrijver~\cite{SchrijverKneser78} as follows.

\begin{theorem}[\cite{LovaszKneser,SchrijverKneser78}]\label{thm:KneserS}
For all integers $n \geq 2k$, $\chi(K(n,k)) = \chi(S(n,k)) = n-2k+2$.
\end{theorem}

The computational search problem associated with the Kneser graph is defined as follows.
\begin{definition}\label{def:KneserSProblems}
In the computational $\KneserP$ problem, the input is a coloring $c: \binom{[n]}{k} \rightarrow [n-2k+1]$ of the vertices of the Kneser graph $K(n,k)$ with $n-2k+1$ colors for integers $n$ and $k$ with $n \geq 2k$, and the goal is to find a monochromatic edge, i.e., $A,B \in \binom{[n]}{k}$ satisfying $A \cap B = \emptyset$ and $c(A) = c(B)$.
In the black-box input model, the input coloring is given as an oracle access that for a vertex $A$ returns its color $c(A)$.
In the white-box input model, the input coloring is given by a Boolean circuit that for a vertex $A$ computes its color $c(A)$.
\end{definition}
\noindent
The existence of a solution to every instance of the $\KneserP$ problem follows from Theorem~\ref{thm:KneserS}.

\subsection{Intersecting Families}
For integers $n$ and $k$ with $n \geq 2k$, let $\calF \subseteq \binom{[n]}{k}$ be a family of $k$-subsets of $[n]$.
We call $\calF$ {\em intersecting} if for every two sets $F_1, F_2 \in \calF$ it holds that $F_1 \cap F_2 \neq \emptyset$.
The Erd{\"{o}}s-Ko-Rado theorem~\cite{EKR61} asserts that every intersecting family $\calF \subseteq \binom{[n]}{k}$ satisfies $|\calF| \leq \binom{n-1}{k-1}$.
This bound is tight and is attained, for each $i \in [n]$, by the family $\{ A \in \binom{[n]}{k} \mid i \in A \}$.
An intersecting family of sets is said to be {\em trivial} if its members share a common element.
Hilton and Milner~\cite{HM67} proved the following stability result for the Erd{\"{o}}s-Ko-Rado theorem, providing an upper bound on the size of any non-trivial intersecting family.
\begin{theorem}[Hilton-Milner Theorem~\cite{HM67}]\label{thm:HM}
For all integers $k \geq 2$ and $n \geq 2k$, every non-trivial intersecting family of $k$-subsets of $[n]$ has size at most $\binom{n-1}{k-1}-\binom{n-k-1}{k-1}+1$.
\end{theorem}
\begin{remark}\label{remark:HM}
The bound given in Theorem~\ref{thm:HM} is tight.
To see this, for an arbitrary $k$-subset $F$ of $[n]$ and for an arbitrary element $i \notin F$, consider the family $\calF = \{ A \in \binom{[n]}{k} \mid A \cap F \neq \emptyset,~i \in A\} \cup \{F\}$.
The family $\calF$ is intersecting and non-trivial, and its size coincides with the bound given in Theorem~\ref{thm:HM}. Note that $|\calF| \leq k \cdot \binom{n-2}{k-2}$, provided that $k \geq 3$.
\end{remark}

\subsection{Chernoff-Hoeffding Bound}
We need the following concentration result (see, e.g.,~\cite[Theorem~2.1]{McDiarmid98}).
\begin{theorem}[Chernoff-Hoeffding Bound]\label{thm:chernoff}
Let $0 < p < 1$, let $X_1,\ldots,X_m$ be $m$ independent binary random variables satisfying $\Prob{}{X_i=1}=p$ and $\Prob{}{X_i=0}=1-p$ for all $i$, and put $\overline{X} = \frac{1}{m} \cdot \sum_{i=1}^{m}{X_i}$. Then, for any $\mu \geq 0$,
\[ \Prob{}{|\overline{X}-p| \geq \mu} \leq 2 \cdot e^{-2m \mu^2}.\]
\end{theorem}

\section{A Fixed-Parameter Algorithm for the $\KneserP$ Problem}\label{sec:algo}

In this section we present and analyze our randomized fixed-parameter algorithm for the $\KneserP$ problem.
We start with a couple of lemmas on induced subgraphs of Kneser graphs that will play a central role in the analysis of the algorithm.
We then describe an algorithm, called `element elimination', which forms a main ingredient in our algorithm for the $\KneserP$ problem, and then use it to present the final algorithm and to prove Theorem~\ref{thm:AlgoKneserNew}.

\subsection{Induced Subgraphs of Kneser Graphs}\label{sec:properties}
The following lemma shows that in a large induced subgraph of $K(n,k)$ whose vertices do not have a popular element, a random pair of vertices forms an edge with a non-negligible probability.
Recall that for a family $\calF \subseteq \binom{[n]}{k}$, we let $K(\calF)$ stand for the subgraph of $K(n,k)$ induced by $\calF$ (see Definition~\ref{def:K(F)}).

\begin{lemma}\label{lemma:at_most_gamma}
For integers $k \geq 3$ and $n \geq 2k$, let $\calF$ be a family of $k$-subsets of $[n]$ of size $|\calF| \geq k^2 \cdot \binom{n-2}{k-2}$ and let $\gamma \in (0,1]$.
Suppose that every element of $[n]$ belongs to at most $\gamma$ fraction of the sets of $\calF$.
Then, the probability that two random sets chosen uniformly and independently from $\calF$ are adjacent in $K(\calF)$ is at least
\[\frac{1}{2} \cdot \Biggr ( 1 - \gamma - \frac{k}{|\calF|} \cdot \binom{n-2}{k-2}\Biggr ) \cdot \Biggr ( 1 - \frac{k^2}{|\calF|} \cdot \binom{n-2}{k-2}\Biggr ).\]
\end{lemma}

\begin{proof}
Let $\calF \subseteq \binom{[n]}{k}$ be a family as in the lemma.
We first claim that every subfamily $\calF' \subseteq \calF$ whose size satisfies
\begin{eqnarray}\label{eq:size_F'}
|\calF'| \geq \gamma \cdot |\calF| + k \cdot \binom{n-2}{k-2}
\end{eqnarray}
spans an edge in $K(\calF)$.
To see this, consider such an $\calF'$, and notice that the assumption that every element of $[n]$ belongs to at most $\gamma$ fraction of the sets of $\calF$, combined with the fact that $|\calF'| > \gamma \cdot |\calF|$, implies that $\calF'$ is not a trivial family, that is, its sets do not share a common element.
In addition, by the Hilton-Milner theorem (Theorem~\ref{thm:HM}; see Remark~\ref{remark:HM}), using $k \geq 3$ and $|\calF'| > k \cdot \binom{n-2}{k-2}$, it follows that $\calF'$ is not a non-trivial intersecting family. We thus conclude that $\calF'$ is not an intersecting family, hence it spans an edge in $K(\calF)$.

We next show a lower bound on the size of a maximum matching in $K(\calF)$.
Consider the process that maintains a subfamily $\calF'$ of $\calF$, initiated as $\calF$, and that removes from $\calF'$ the two endpoints of some edge spanned by $\calF'$ as long as its size satisfies the condition given in~\eqref{eq:size_F'}. The pairs of vertices that are removed during the process form a matching $\calM$ in $K(\calF)$, whose size satisfies
\begin{eqnarray}\label{eq:matching}
|\calM| &\geq& \frac{1}{2} \cdot \Biggr (  |\calF| -  \Biggr (\gamma \cdot |\calF| + k \cdot \binom{n-2}{k-2} \Biggr ) \Biggr ) = \frac{1}{2} \cdot \Biggr ( (1-\gamma) \cdot |\calF| - k \cdot \binom{n-2}{k-2}\Biggr ).
\end{eqnarray}

We proceed by considering the sum of the degrees of adjacent vertices in the graph $K(\calF)$.
Let $A,B \in \calF$ be any adjacent vertices in $K(\calF)$.
Since $A$ and $B$ are adjacent, they satisfy $A \cap B = \emptyset$, hence every vertex of $\calF$ that is not adjacent to $A$ nor to $B$ must include two distinct elements $j_A \in A$ and $j_B \in B$. It thus follows that the number of vertices of $\calF$ that are not adjacent to $A$ nor to $B$ does not exceed $k^2 \cdot \binom{n-2}{k-2}$. This implies that the degrees of $A$ and $B$ in $K(\calF)$ satisfy
\[d(A)+d(B) \geq |\calF| - k^2 \cdot \binom{n-2}{k-2}.\]

Let $\calE$ denote the edge set of $K(\calF)$. We combine the above bound with the lower bound given in~\eqref{eq:matching} on the size of the matching $\calM$, to obtain that
\begin{eqnarray*}
2 \cdot |\calE| = \sum_{F \in \calF}{d(F)} &\geq&
\sum_{\{A,B\} \in \calM}{(d(A)+d(B))} \geq
|\calM| \cdot \Biggr (  |\calF| - k^2 \cdot \binom{n-2}{k-2}\Biggr )
\\ &\geq& \frac{1}{2} \cdot \Biggr ( (1-\gamma) \cdot |\calF| - k \cdot \binom{n-2}{k-2}\Biggr ) \cdot \Biggr (  |\calF| - k^2 \cdot \binom{n-2}{k-2}\Biggr ).
\end{eqnarray*}
Finally, consider a pair of random vertices chosen uniformly and independently from $\calF$.
The probability that they form an edge in $K(\calF)$ is twice the number of edges in $K(\calF)$ divided by $|\calF|^2$.
Hence, the above bound on $2 \cdot |\calE|$ completes the proof.
\end{proof}

As a corollary of Lemma~\ref{lemma:at_most_gamma}, we obtain the following.
\begin{corollary}\label{cor:at_most_gamma}
For integers $k \geq 3$ and $n \geq 8k^4$, let $\calF$ be a family of $k$-subsets of $[n]$ of size $|\calF| \geq \frac{1}{2n} \cdot \binom{n}{k}$ and let $\gamma \in (0,1]$.
Suppose that every element of $[n]$ belongs to at most $\gamma$ fraction of the sets of $\calF$.
Then, the probability that two random sets chosen uniformly and independently from $\calF$ are adjacent in $K(\calF)$ is at least $\frac{3}{8} \cdot (\frac{3}{4}-\gamma)$.
\end{corollary}

\begin{proof}
Observe that the assumptions $|\calF| \geq \frac{1}{2n} \cdot \binom{n}{k}$ and $n \geq 8k^4$ imply that
\[ \frac{k^2}{|\calF|} \cdot \binom{n-2}{k-2} \leq \frac{2n k^2}{\binom{n}{k}} \cdot\binom{n-2}{k-2} = \frac{2k^3 (k-1)}{n-1} \leq \frac{1}{4}.\]
Applying Lemma~\ref{lemma:at_most_gamma}, we obtain that the probability that two random sets chosen uniformly and independently from $\calF$ are adjacent in $K(\calF)$ is at least
$\frac{1}{2} \cdot (1-\gamma-\frac{1}{4}) \cdot (1-\frac{1}{4}) = \frac{3}{8} \cdot (\frac{3}{4}-\gamma)$.
\end{proof}

The following lemma shows that if a large collection of vertices of $K(n,k)$ has a quite popular element, then every vertex that does not include this element is adjacent to many of the vertices in the collection.

\begin{lemma}\label{lemma:at_least_gamma}
For integers $k \geq 2$ and $n \geq 2k$, let $X \subseteq [n]$ be a set, let $\calF$ be a family of $k$-subsets of $X$, and let $\gamma \in (0,1]$.
Let $j \in X$ be an element that belongs to at least $\gamma$ fraction of the sets of $\calF$, and suppose that $A \in \binom{[n]}{k}$ is a set satisfying $j \notin A$.
Then, the probability that a random set chosen uniformly from $\calF$ is disjoint from $A$ is at least
\[\gamma - \frac{k}{|\calF|} \cdot \binom{|X|-2}{k-2}.\]
\end{lemma}

\begin{proof}
Let $\calF \subseteq \binom{X}{k}$ be a family as in the lemma, and put $\calF' = \{F \in \calF \mid j \in F\}$.
By assumption, it holds that $|\calF'| \geq \gamma \cdot |\calF|$.
Suppose that $A \in \binom{[n]}{k}$ is a set satisfying $j \notin A$.
Observe that the number of sets $B \in \binom{X}{k}$ with $j \in B$ that intersect $A$ does not exceed $|A \cap X| \cdot \binom{|X|-2}{k-2} \leq k \cdot \binom{|X|-2}{k-2}$.
It thus follows that the number of sets of $\calF$ that are disjoint from $A$ is at least
\[ |\calF'| - k \cdot \binom{|X|-2}{k-2} \geq \gamma \cdot |\calF| -k \cdot \binom{|X|-2}{k-2}.\]
Hence, a random set chosen uniformly from $\calF$ is disjoint from $A$ with the desired probability.
\end{proof}

As a corollary of Lemma~\ref{lemma:at_least_gamma}, we obtain the following.

\begin{corollary}\label{cor:at_least_gamma}
For integers $k \geq 2$ and $n$, let $X \subseteq [n]$ be a set of size $|X| \geq 16k^3$, let $\calF$ be a family of $k$-subsets of $X$ of size $|\calF| \geq \frac{1}{2|X|} \cdot \binom{|X|}{k}$, and let $\gamma \in (0,1]$.
Let $j \in X$ be an element that belongs to at least $\gamma$ fraction of the sets of $\calF$, and suppose that $A \in \binom{[n]}{k}$ is a set satisfying $j \notin A$.
Then, the probability that a random set chosen uniformly from $\calF$ is disjoint from $A$ is at least $\gamma-\frac{1}{8}$.
\end{corollary}

\begin{proof}
Observe that the assumptions $|\calF| \geq \frac{1}{2|X|} \cdot \binom{|X|}{k}$ and $|X| \geq 16 k^3$ imply that
\[ \frac{k}{|\calF|} \cdot \binom{|X|-2}{k-2} \leq \frac{2|X| k}{\binom{|X|}{k}} \cdot\binom{|X|-2}{k-2} = \frac{2k^2 (k-1)}{|X|-1} \leq \frac{1}{8}.\]
Applying Lemma~\ref{lemma:at_least_gamma}, we obtain that the probability that a random set chosen uniformly from $\calF$ is disjoint from $A$ is at least $\gamma-\frac{1}{8}$.
\end{proof}

\subsection{The Element Elimination Algorithm}

A main ingredient in our fixed-parameter algorithm for the $\KneserP$ problem is the `element elimination' algorithm given by the following theorem.
It will be used to repeatedly reduce the size of the ground set of a Kneser graph while looking for a monochromatic edge.

\begin{theorem}\label{thm:one_step}
There exists a randomized algorithm that given integers $n$ and $k$, a set $X \subseteq [n]$ of size $|X| \geq 8 k^4$, a set of colors $C \subseteq [n-2k+1]$ of size $|C|=|X|-2k+1$, and an oracle access to a coloring $c: \binom{X}{k} \rightarrow [n-2k+1]$ of the vertices of $K(\binom{X}{k})$, runs in time polynomial in $n$ and returns, with probability $1-2^{-\Omega(n)}$,
\begin{enumerate}[(a).]
  \item\label{output:a} a monochromatic edge of $K(\binom{X}{k})$, or
  \item\label{output:b} a vertex $A \in \binom{X}{k}$ satisfying $c(A) \notin C$, or
  \item\label{output:c} a color $i \in C$ and an element $j \in X$ such that for every $A \in \binom{[n]}{k}$ with $j \notin A$, a random vertex $B$ chosen uniformly from $\binom{X}{k}$ satisfies $c(B)=i$ and $A \cap B = \emptyset$ with probability at least $\frac{1}{9n}$.
\end{enumerate}
\end{theorem}

\begin{proof}
For integers $n$ and $k$, let $X \subseteq [n]$, $C \subseteq [n-2k+1]$, and $c: \binom{X}{k} \rightarrow [n-2k+1]$ be an input satisfying $|X| \geq 8 k^4$ and $|C|=|X|-2k+1$ as in the statement of the theorem.
It can be assumed that $k \geq 3$. Indeed, Theorem~\ref{thm:KneserS} guarantees that $K(\binom{X}{k})$ has either a monochromatic edge or a vertex whose color does not belong to $C$, hence for $k \leq 2$, an output of type~\eqref{output:a} or~\eqref{output:b} can be found by querying the oracle for the colors of all the vertices in time polynomial in $n$.
For $k \geq 3$, consider the algorithm that given an input as above acts as follows (see Algorithm~\ref{alg:elimination}).

The algorithm first selects uniformly and independently $m$ random sets $A_1, \ldots, A_m \in \binom{X}{k}$ for $m=n^3$ (see lines~\ref{line:pick_s}--\ref{line:samples}) and queries the oracle for their colors.
If the sampled sets include two vertices that form a monochromatic edge, then the algorithm returns such an edge (output of type~\eqref{output:a}; see line~\ref{line:output_a}).
If they include a vertex whose color does not belong to $C$, then the algorithm returns it (output of type~\eqref{output:b}; see line~\ref{line:output_b}).
Otherwise, the algorithm defines $i^* \in C$ as a color that appears on a largest number of sampled sets $A_t$ (see lines~\ref{line:alpha_s}--\ref{line:alpha_max}).
It further defines $j^* \in X$ as an element that belongs to a largest number of sampled sets $A_t$ with $c(A_t)=i^*$ (see lines~\ref{line:gamma_s}--\ref{line:gamma_max}).
Then, the algorithm returns the pair $(i^*,j^*)$ (output of type~\eqref{output:c}; see line~\ref{line:output_c}).

\begin{algorithm}[ht]
    \caption{Element Elimination Algorithm (Theorem~\ref{thm:one_step})}
    \textbf{Input:} Integers $n$ and $k \geq 3$, a set $X \subseteq [n]$ of size $|X| \geq 8 k^4$, a set of colors $C \subseteq [n-2k+1]$ of size $|C|=|X|-2k+1$, and an oracle access to a coloring $c: \binom{X}{k} \rightarrow [n-2k+1]$ of $K(\binom{X}{k})$.\\
    \textbf{Output:} \eqref{output:a} A monochromatic edge of $K(\binom{X}{k})$, or \eqref{output:b} a vertex $A \in \binom{X}{k}$ satisfying $c(A) \notin C$, or \eqref{output:c} a color $i \in C$ and an element $j \in X$ such that for every $A \in \binom{[n]}{k}$ with $j \notin A$, a random vertex $B$ chosen uniformly from $\binom{X}{k}$ satisfies $c(B)=i$ and $A \cap B = \emptyset$ with probability at least $\frac{1}{9n}$.
    \begin{algorithmic}[1]
        \State{$m \leftarrow n^3$}\label{line:pick_s}
        \State{pick uniformly and independently at random sets $A_1, \ldots, A_m \in \binom{X}{k}$}\label{line:samples}
        \ForAll{$t,t' \in [m]$}
            \If{$c(A_t) =c(A_{t'})$ and $A_t \cap A_{t'} = \emptyset$}
                \State{\textbf{return} $\{A_t,A_{t'}\}$}\label{line:output_a}\Comment{output of type~\eqref{output:a}}
            \EndIf
        \EndFor
        \ForAll{$t \in [m]$}
            \If{$c(A_t) \notin C$}
                \State{\textbf{return} $A_t$}\label{line:output_b}\Comment{output of type~\eqref{output:b}}
            \EndIf
        \EndFor
        \ForAll{$i \in C$}\label{line:alpha_s}
            \State{$\widetilde{\alpha}_i \leftarrow \frac{1}{m} \cdot  |\{ t \in [m] \mid c(A_t)=i\}|$}\label{line:alpha}
        \EndFor
        \State{$i^* \leftarrow $~an $i \in C$ with largest value of $\widetilde{\alpha}_{i}$}\label{line:alpha_max}
        \ForAll{$j \in X$}\label{line:gamma_s}
            \State{$\widetilde{\gamma}_{i^*,j} \leftarrow \frac{1}{m} \cdot  |\{ t \in [m] \mid c(A_t)=i^*~\mbox{and}~j \in A_t\}|$}\label{line:gamma}
        \EndFor
        \State{$j^* \leftarrow $~a $j \in X$ with largest value of $\widetilde{\gamma}_{i^*,j}$}\label{line:gamma_max}
        \State{\textbf{return} $(i^*,j^*)$}\label{line:output_c}\Comment{output of type~\eqref{output:c}}
    \end{algorithmic}
    \label{alg:elimination}
\end{algorithm}

It is clear that the algorithm runs in time polynomial in $n$.
We turn to prove that for every input, the algorithm returns a valid output, of type~\eqref{output:a},~\eqref{output:b}, or~\eqref{output:c}, with probability $1-2^{-\Omega(n)}$.
We start with the following lemma that shows that if the input coloring has a large color class with no popular element, then with high probability the algorithm returns a valid output of type~\eqref{output:a}.

\begin{lemma}\label{lemma:elimination_no_pop}
Suppose that the input coloring $c$ has a color class $\calF \subseteq \binom{X}{k}$ of size $|\calF| \geq \frac{1}{2|X|} \cdot \binom{|X|}{k}$ such that every element of $X$ belongs to at most half of the sets of $\calF$.
Then, Algorithm~\ref{alg:elimination} returns a monochromatic edge with probability $1-2^{-\Omega(n)}$.
\end{lemma}

\begin{proof}
Let $\calF$ be as in the lemma.
Applying Corollary~\ref{cor:at_most_gamma} with $\gamma = \frac{1}{2}$, using the assumptions $k \geq 3$ and $|X| \geq 8k^4$, we obtain that two random sets chosen uniformly and independently from $\calF$ are adjacent in $K(\calF)$ with probability at least $\frac{3}{8} \cdot (\frac{3}{4}-\gamma) = \frac{3}{32}$.
Further, the fact that $|\calF| \geq \frac{1}{2|X|} \cdot \binom{|X|}{k}$ implies that a random vertex chosen uniformly from $\binom{X}{k}$ belongs to $\calF$ with probability at least $\frac{1}{2|X|}$. Hence, for two random vertices chosen uniformly and independently from $\binom{X}{k}$, the probability that they both belong to $\calF$ is at least $(\frac{1}{2|X|})^2$, and conditioned on this event, their probability to form an edge in $K(\calF)$ is at least $\frac{3}{32}$. This implies that the probability that two random vertices chosen uniformly and independently from $\binom{X}{k}$ form a monochromatic edge in $K(\binom{X}{k})$ is at least $(\frac{1}{2|X|})^2 \cdot \frac{3}{32} = \frac{3}{128|X|^2}$.

Now, by considering $\lfloor m/2 \rfloor$ pairs of the random sets chosen by Algorithm~\ref{alg:elimination} (line~\ref{line:samples}), it follows that the probability that no pair forms a monochromatic edge does not exceed
\[ \Big (1-\tfrac{3}{128|X|^2} \Big )^{\lfloor m/2\rfloor } \leq e^{-3 \cdot \lfloor m/2\rfloor / (128|X|^2)} \leq 2^{-\Omega(n)},\]
where the last inequality follows by $|X| \leq n$ and $m=n^3$.
It thus follows that with probability $1-2^{-\Omega(n)}$, the algorithm returns a monochromatic edge, as required.
\end{proof}

We next handle the case in which every large color class of the input coloring has a popular element.
To do so, we first show that the samples of the algorithm provide a good estimation for the fraction of vertices in each color class as well as for the fraction of vertices that share any given element in each color class.
For every color $i \in C$, let $\alpha_i$ denote the fraction of vertices of $K(\binom{X}{k})$ colored $i$, that is,
\[\alpha_i = \frac{|\{ A \in \binom{X}{k} \mid c(A)=i\}|}{\binom{|X|}{k}},\]
and let $\widetilde{\alpha}_i$ denote the fraction of the vertices sampled by the algorithm that are colored $i$ (see line~\ref{line:alpha}).
Similarly, for every $i \in C$ and $j \in X$, let $\gamma_{i,j}$ denote the fraction of vertices of $K(\binom{X}{k})$ colored $i$ that include $j$, that is,
\[\gamma_{i,j} = \frac{|\{ A \in \binom{X}{k} \mid c(A)=i~\mbox{and}~j \in A\}|}{\binom{|X|}{k}},\]
and let $\widetilde{\gamma}_{i,j}$ denote the fraction of the vertices sampled by the algorithm that are colored $i$ and include $j$.
Let $E$ denote the event that
\begin{eqnarray}\label{eq:approx}
|\alpha_i - \widetilde{\alpha}_i| \leq \frac{1}{9|X|}~~\mbox{and}~~|\gamma_{i,j} - \widetilde{\gamma}_{i,j}| \leq \frac{1}{9|X|}~~\mbox{for all}~i \in C,~j \in X.
\end{eqnarray}

By a standard concentration argument, we obtain the following lemma.
\begin{lemma}\label{lemma:E}
The probability of the event $E$ is $1-2^{-\Omega(n)}$.
\end{lemma}
\begin{proof}
By the Chernoff-Hoeffding bound (Theorem~\ref{thm:chernoff}) applied with $\mu=\frac{1}{9|X|}$, the probability that an inequality from~\eqref{eq:approx} does not hold is at most
\[2 \cdot e^{-2m/(81|X|^2)} \leq 2 \cdot e^{-n/2},\]
where the inequality follows by $|X| \leq n$ and $m=n^3$.
By the union bound over all the colors $i \in C$ and all the pairs $(i,j) \in C \times X$, that is, over $|C| \cdot (1+|X|) \leq n^2$ events, we get that all the inequalities in~\eqref{eq:approx} hold with probability at least $1 - 2n^2 \cdot e^{-n/2} = 1-2^{-\Omega(n)}$, as required.
\end{proof}

We now show that if every large color class of the input coloring has a popular element and the event $E$ occurs, then the algorithm returns a valid output.

\begin{lemma}\label{lemma:elimination_pop}
Suppose that the coloring $c$ satisfies that for every color class $\calF \subseteq \binom{X}{k}$ of size $|\calF| \geq \frac{1}{2|X|} \cdot \binom{|X|}{k}$ there exists an element of $X$ that belongs to more than half of the sets of $\calF$.
Then, if the event $E$ occurs, Algorithm~\ref{alg:elimination} returns a valid output.
\end{lemma}

\begin{proof}
Assume that the event $E$ occurs.
If Algorithm~\ref{alg:elimination} returns an output of type~\eqref{output:a} or~\eqref{output:b}, i.e., a monochromatic edge or a vertex whose color does not belong to $C$, then the output is verified before it is returned and is thus valid.
So suppose that the algorithm returns a pair $(i^*,j^*) \in C \times X$.
Recall that the color $i^*$ is defined by Algorithm~\ref{alg:elimination} as an $i \in C$ with largest value of $\widetilde{\alpha}_i$ (see line~\ref{line:alpha_max}).
Since the colors of all the sampled sets belong to $C$, it follows that $\sum_{i \in C}{\widetilde{\alpha}_i} = 1$, and thus
\begin{eqnarray}\label{eq:alpha_i_star}
\widetilde{\alpha}_{i^*} \geq \frac{1}{|C|} \geq \frac{1}{|X|},
\end{eqnarray}
where the last inequality follows by $|C| = |X|-2k+1 \leq |X|$.

Let $\calF$ be the family of vertices of $K(\binom{X}{k})$ colored $i^*$, i.e., $\calF = \{A \in \binom{X}{k} \mid c(A)=i^*\}$.
Since the event $E$ occurs (see~\eqref{eq:approx}), it follows from~\eqref{eq:alpha_i_star} that
\begin{eqnarray}\label{eq:alpha_i_new}
\alpha_{i^*} \geq \widetilde{\alpha}_{i^*} - \frac{1}{9|X|} \geq \frac{1}{|X|} - \frac{1}{9|X|} = \frac{8}{9|X|}
\end{eqnarray}
and thus
\[|\calF| = \alpha_{i^*} \cdot \binom{|X|}{k} \geq \frac{8}{9|X|}\cdot \binom{|X|}{k}.\]
Hence, by the assumption of the lemma, there exists an element $j \in X$ that belongs to more than half of the sets of $\calF$, that is, $\gamma_{i^*,j} > \frac{1}{2} \cdot \alpha_{i^*}$.
Since the event $E$ occurs, it follows that this $j$ satisfies $\widetilde{\gamma}_{i^*,j} > \frac{1}{2} \cdot \alpha_{i^*}- \frac{1}{9|X|}$.
Recalling that the element $j^*$ is defined by Algorithm~\ref{alg:elimination} as a $j \in X$ with largest value of $\widetilde{\gamma}_{i^*,j}$ (see line~\ref{line:gamma_max}), it must satisfy $\widetilde{\gamma}_{i^*,j^*} > \frac{1}{2} \cdot \alpha_{i^*} - \frac{1}{9|X|}$, and using again the fact that the event $E$ occurs, we derive that $\gamma_{i^*,j^*} \geq \widetilde{\gamma}_{i^*,j^*} -\frac{1}{9|X|} > \frac{1}{2} \cdot \alpha_{i^*} - \frac{2}{9|X|}$.
This implies that $j^*$ belongs to at least $\gamma$ fraction of the sets of $\calF$ for $\gamma = \frac{\gamma_{i^*,j^*}}{\alpha_{i^*}} > \frac{1}{2} - \frac{2}{9 |X| \cdot \alpha_{i^*}} \geq \frac{1}{4}$, where the last inequality follows from~\eqref{eq:alpha_i_new}.

By $k \geq 3$ and $|X| \geq 8k^4 \geq 8k^3$, we can apply Corollary~\ref{cor:at_least_gamma} with $\calF$, $j^*$, and $\gamma \geq \frac{1}{4}$ to obtain that for every set $A \in \binom{[n]}{k}$ with $j^* \notin A$, the probability that a random set chosen uniformly from $\calF$ is disjoint from $A$ is at least $\gamma-\frac{1}{8} \geq \frac{1}{8}$. Since the probability that a random set chosen uniformly from $\binom{X}{k}$ belongs to $\calF$ is at least $\frac{8}{9|X|}$, it follows that the probability that a random set $B$ chosen uniformly from $\binom{X}{k}$ satisfies $c(B) = i^*$ and $A \cap B = \emptyset$ is at least $\frac{8}{9|X|} \cdot \frac{1}{8} = \frac{1}{9|X|} \geq \frac{1}{9n}$. This implies that $(i^*,j^*)$ is a valid output of type~\eqref{output:c}, as required.
\end{proof}

Equipped with Lemmas~\ref{lemma:elimination_no_pop},~\ref{lemma:E}, and~\ref{lemma:elimination_pop}, we are ready to derive the correctness of Algorithm~\ref{alg:elimination} and to complete the proof of Theorem~\ref{thm:one_step}.
If the input coloring $c$ has a color class $\calF \subseteq \binom{X}{k}$ of size $|\calF| \geq \frac{1}{2|X|} \cdot \binom{|X|}{k}$ such that every element of $X$ belongs to at most half of the sets of $\calF$, then, by Lemma~\ref{lemma:elimination_no_pop}, the algorithm returns with probability $1-2^{-\Omega(n)}$ a monochromatic edge, i.e., a valid output of type~\eqref{output:a}.
Otherwise, the input coloring $c$ satisfies that for every color class $\calF \subseteq \binom{X}{k}$ of size $|\calF| \geq \frac{1}{2|X|} \cdot \binom{|X|}{k}$ there exists an element of $X$ that belongs to more than half of the sets of $\calF$. By Lemma~\ref{lemma:E}, the event $E$ occurs with probability $1-2^{-\Omega(n)}$, implying by Lemma~\ref{lemma:elimination_pop} that with such probability, the algorithm returns a valid output.
It thus follows that for every input coloring the algorithm returns a valid output with probability $1-2^{-\Omega(n)}$, and we are done.
\end{proof}

\subsection{The Fixed-Parameter Algorithm for the $\KneserP$ Problem}

We turn to present our fixed-parameter algorithm for the $\KneserP$ problem and to complete the proof of Theorem~\ref{thm:AlgoKneserNew}.

\begin{proof}[ of Theorem~\ref{thm:AlgoKneserNew}]
Suppose that we are given, for integers $n \geq 2k$, an oracle access to a coloring $c: \binom{[n]}{k} \rightarrow [n-2k+1]$ of the vertices of the Kneser graph $K(n,k)$.
Our algorithm has two phases, as described below (see Algorithm~\ref{alg:Kneser}).

\begin{algorithm}[!htp]
    \caption{The Algorithm for the $\KneserP$ Problem (Theorem~\ref{thm:AlgoKneserNew})}
    \textbf{Input:} Integers $n,k$ with $n \geq 2k$ and an oracle access to a coloring $c: \binom{[n]}{k} \rightarrow [n-2k+1]$ of $K(n,k)$.\\
    \textbf{Output:} A monochromatic edge of $K(n,k)$.
    \begin{algorithmic}[1]
        \State{$s \leftarrow \max (n-8k^4,0)$,~~$X_0 \leftarrow [n]$,~~$C_0 \leftarrow [n-2k+1]$}\Comment{$|C_0|=|X_0|-2k+1$}
        \ForAll{$l = 0,1,\ldots,s-1$}\label{line:first_p}\Comment{first phase}
            \State{\textbf{call} Algorithm~\ref{alg:elimination} with $n$, $k$, $X_l$, $C_l$, and with the restriction of $c$ to $\binom{X_l}{k}$}
            \If{Algorithm~\ref{alg:elimination} returns an edge $\{A,B\}$ with $c(A)=c(B)$}\label{line:a_s}\Comment{output of type~\eqref{output:a}}
                \State{\textbf{return} $\{A,B\}$}
            \EndIf\label{line:a_e}
            \If{Algorithm~\ref{alg:elimination} returns a vertex $A \in \binom{X_l}{k}$ with $c(A) = i_r \notin C_l$}\label{line:b_s}\label{line:B_t_1}\Comment{output of type~\eqref{output:b}}
                \ForAll{$t \in [n^2]$}
                    \State{pick uniformly at random a set $B_t \in \binom{X_r}{k}$}
                    \If{$c(B_t) = i_r$ and $A \cap B_t = \emptyset$}
                        \State{\textbf{return} $\{A,B_t\}$}
                    \EndIf
                \EndFor
                \State{\textbf{return} `failure'}
            \EndIf\label{line:b_e}
            \If{Algorithm~\ref{alg:elimination} returns a pair $(i_l,j_l) \in C_l \times X_l$}\label{line:c_s}\Comment{output of type~\eqref{output:c}}
                \State{$X_{l+1} \leftarrow X_l \setminus \{j_l\},~~C_{l+1} \leftarrow C_l \setminus \{i_l\}$}\Comment{$|C_{l+1}|=|X_{l+1}|-2k+1$}
            \EndIf\label{line:c_e}
        \EndFor\label{line:first_p_end}
        \State{query the oracle for the colors of all the vertices of $K(\binom{X_s}{k})$}\Comment{second phase}
        \If{there exists a vertex $A \in \binom{X_s}{k}$ of color $c(A) = i_r \notin C_s$}\label{line:bb_s}
            \ForAll{$t \in [n^2]$}
                \State{pick uniformly at random a set $B_t \in \binom{X_r}{k}$}
                \If{$c(B_t) = i_r$ and $A \cap B_t = \emptyset$}\label{line:B_t_2}
                    \State{\textbf{return} $\{A,B_t\}$}
                \EndIf
            \EndFor
            \State{\textbf{return} `failure'}\label{line:bb_e}
        \Else \State{find $A,B \in \binom{X_s}{k}$ satisfying $c(A)=c(B)$ and $A \cap B = \emptyset$}\label{line:kneser_s}\Comment{exist by Theorem~\ref{thm:KneserS}~\cite{LovaszKneser}}
            \State{\textbf{return} $\{A,B\}$}\label{line:kneser_e}
        \EndIf
    \end{algorithmic}
    \label{alg:Kneser}
\end{algorithm}

In the first phase, the algorithm repeatedly applies the `element elimination' algorithm given in Theorem~\ref{thm:one_step} (Algorithm~\ref{alg:elimination}).
Initially, we define
\[s=\max(n-8k^4,0),~~X_0 = [n],~~\mbox{and}~~C_0 = [n-2k+1].\]
In the $l$th iteration, $0 \leq l < s$, we call Algorithm~\ref{alg:elimination} with $n$, $k$, $X_l$, $C_l$, and with the restriction of the given coloring $c$ to the vertices of $\binom{X_l}{k}$ to obtain with probability $1-2^{-\Omega(n)}$,
\begin{enumerate}[(a).]
    \item\label{output:a1} a monochromatic edge $\{A,B\}$ of $K(\binom{X_l}{k})$, or
    \item\label{output:b1} a vertex $A \in \binom{X_l}{k}$ satisfying $c(A) \notin C_l$, or
    \item\label{output:c1} a color $i_l \in C_l$ and an element $j_l \in X_l$ such that for every $A \in \binom{[n]}{k}$ with $j_l \notin A$, a random vertex $B$ chosen uniformly from $\binom{X_l}{k}$ satisfies $c(B)=i_l$ and $A \cap B = \emptyset$ with probability at least $\frac{1}{9n}$.
\end{enumerate}
As will be explained shortly, if the output of Algorithm~\ref{alg:elimination} is of type~\eqref{output:a1} or~\eqref{output:b1} then we either return a monochromatic edge or declare `failure', and if the output is a pair $(i_l,j_l)$ of type~\eqref{output:c1} then we define $X_{l+1} = X_l \setminus \{j_l\}$ and $C_{l+1} = C_l \setminus \{i_l\}$ and, as long as $l < s$, proceed to the next call of Algorithm~\ref{alg:elimination}. Note that the sizes of the sets $X_l$ and $C_l$ are reduced by $1$ in every iteration, hence we maintain the equality $|C_l| = |X_l| - 2k+1$ for all $l$.
We now describe how the algorithm acts in the $l$th iteration for each type of output returned by Algorithm~\ref{alg:elimination}.

If the output is of type~\eqref{output:a1}, then the returned monochromatic edge of $K(\binom{X_l}{k})$ is also a monochromatic edge of $K(n,k)$, so we return it (see lines~\ref{line:a_s}--\ref{line:a_e}).

If the output is of type~\eqref{output:b1}, then we are given a vertex $A \in \binom{X_l}{k}$ satisfying $c(A) = i_r \notin C_l$ for some $r<l$. Since $i_r \notin C_l$, it follows that $j_r \notin X_l$, and thus $j_r \notin A$.
In this case, we pick uniformly and independently $n^2$ random sets from $\binom{X_r}{k}$ and query the oracle for their colors.
If we find a vertex $B$ that forms together with $A$ a monochromatic edge in $K(n,k)$, we return the monochromatic edge $\{A,B\}$, and otherwise we declare `failure' (see lines~\ref{line:b_s}--\ref{line:b_e}).

If the output of Algorithm~\ref{alg:elimination} is a pair $(i_l,j_l)$ of type~\eqref{output:c1}, then we define, as mentioned above, the sets $X_{l+1} = X_l \setminus \{j_l\}$ and $C_{l+1} = C_l \setminus \{i_l\}$ (see lines~\ref{line:c_s}--\ref{line:c_e}).
Observe that for $0 \leq l < s$, it holds that $|X_{l}| = n-l > n-s = 8k^4$, allowing us, by Theorem~\ref{thm:one_step}, to call Algorithm~\ref{alg:elimination} in the $l$th iteration.

In case that all the $s$ calls to Algorithm~\ref{alg:elimination} return an output of type~\eqref{output:c1}, we arrive to the second phase of the algorithm.
Here, we are given the sets $X_s$ and $C_s$ that satisfy $|X_s| = n-s \leq 8k^4$ and $|C_s| = |X_s|-2k+1$, and we query the oracle for the colors of each and every vertex of the graph $K(\binom{X_s}{k})$.
If we find a vertex $A \in \binom{X_s}{k}$ satisfying $c(A) = i_r \notin C_s$ for some $r<s$, then, as before, we pick uniformly and independently $n^2$ random sets from $\binom{X_r}{k}$ and query the oracle for their colors. If we find a vertex $B$ that forms together with $A$ a monochromatic edge in $K(n,k)$, we return the monochromatic edge $\{A,B\}$, and otherwise we declare `failure' (see lines~\ref{line:bb_s}--\ref{line:bb_e}).
Otherwise, all the vertices of $K(\binom{X_s}{k})$ are colored by colors from $C_s$. By Theorem~\ref{thm:KneserS}, the chromatic number of $K(\binom{X_s}{k})$ is $|X_s|-2k+2 > |C_s|$. Hence, there must exist a monochromatic edge in $K(\binom{X_s}{k})$, and by checking all the pairs of its vertices we find such an edge and return it (see lines~\ref{line:kneser_s}--\ref{line:kneser_e}).

We turn to analyze the probability that Algorithm~\ref{alg:Kneser} returns a monochromatic edge.
Note that whenever the algorithm returns an edge, it checks that it is monochromatic and thus ensures that it forms a valid solution.
Hence, it suffices to show that the algorithm declares `failure' with probability $2^{-\Omega(n)}$.
To see this, recall that the algorithm calls Algorithm~\ref{alg:elimination} at most $s < n$ times, and that by Theorem~\ref{thm:one_step} the probability that its output is not valid is $2^{-\Omega(n)}$. By the union bound, the probability that any of the calls to Algorithm~\ref{alg:elimination} returns an invalid output is $2^{-\Omega(n)}$ too.
The only situation in which Algorithm~\ref{alg:Kneser} declares `failure' is when it finds, for some $r <s$, a vertex $A \in \binom{[n]}{k}$ with $c(A)=i_r$ and $j_r \notin A$, and none of the $n^2$ sampled sets $B \in \binom{X_r}{k}$ satisfies $c(B)=i_r$ and $A \cap B = \emptyset$ (see lines~\ref{line:b_s}--\ref{line:b_e},~\ref{line:bb_s}--\ref{line:bb_e}).
However, assuming that all the calls to Algorithm~\ref{alg:elimination} return valid outputs, the $r$th run guarantees, by Theorem~\ref{thm:one_step}, that a random vertex $B$ uniformly chosen from $\binom{X_r}{k}$ satisfies $c(B)=i_r$ and $A \cap B = \emptyset$ for the given $A$ with probability at least $\frac{1}{9n}$.
Hence, the probability that the algorithm declares `failure' does not exceed $(1-\frac{1}{9n})^{n^2} \leq e^{-n/9} =  2^{-\Omega(n)}$.
Using again the union bound, it follows that the probability that Algorithm~\ref{alg:Kneser} either gets an invalid output from Algorithm~\ref{alg:elimination} or fails to find a vertex that forms a monochromatic edge with a set $A$ as above is $2^{-\Omega(n)}$.
Therefore, the probability that Algorithm~\ref{alg:Kneser} successfully finds a monochromatic edge is $1-2^{-\Omega(n)}$, as desired.

We finally analyze the running time of Algorithm~\ref{alg:Kneser}.
In its first phase, the algorithm calls Algorithm~\ref{alg:elimination} at most $s < n$ times, where the running time needed for each call is, by Theorem~\ref{thm:one_step}, polynomial in $n$. It is clear that the other operations made throughout this phase can also be implemented in time polynomial in $n$.
In its second phase, the algorithm enumerates all the vertices of $K(\binom{X_s}{k})$. This phase can be implemented in running time polynomial in $n$ and in the number of vertices of this graph. The latter is $\binom{|X_s|}{k} \leq |X_s|^k \leq (8k^4)^k =  k^{O(k)}$.
It thus follows that the total running time of Algorithm~\ref{alg:Kneser} is $n^{O(1)} \cdot k^{O(k)}$, completing the proof.
\end{proof}

\subsection{Turing Kernelization for the $\KneserP$ Problem}\label{sec:kernel}

As mentioned in the introduction, our algorithm for the $\KneserP$ problem can be viewed as a randomized polynomial Turing kernelization algorithm for the problem.
In what follows we extend on this aspect of the algorithm.

We start with the definition of a Turing kernelization algorithm as it is used in the standard context of parameterized complexity (see, e.g.,~\cite[Chapter~22]{KernelBook19}).
A decision parameterized problem $P$ is a language of pairs $(x,k)$ where $k$ is an integer referred to as the parameter of $P$.
For a decision parameterized problem $P$ and a computable function $f$, a Turing kernelization algorithm of size $f$ for $P$ is an algorithm that decides whether an input $(x,k)$ belongs to $P$ in polynomial time given an oracle access that decides membership in $P$ for instances $(x',k')$ with $|x'| \leq f(k)$ and $k' \leq f(k)$. The kernelization algorithm is said to be polynomial if $f$ is a polynomial function.

The $\KneserP$ problem, however, is a total search problem whose input is given as an oracle access.
Hence, for an instance of the $\KneserP$ problem associated with integers $n$ and $k$, we require a Turing kernelization algorithm to find a solution using an oracle that finds solutions for instances associated with integers $n'$ and $k'$ that are bounded by a function of $k$. We further require the algorithm to be able to simulate the queries of the produced instances using queries to the oracle associated with the original input.
Note that we have here two different types of oracles: the oracle that solves the $\KneserP$ problem on graphs $K(n',k')$ with bounded $n',k'$ and the oracle that supplies an access to the instance of the $\KneserP$ problem.

We claim that the proof of Theorem~\ref{thm:AlgoKneserNew} shows that the $\KneserP$ problem admits a randomized polynomial Turing kernelization algorithm in the following manner.
Given an instance of the $\KneserP$ problem, a coloring $c: \binom{[n]}{k} \rightarrow [n-2k+1]$ of the vertices of a Kneser graph $K(n,k)$ for integers $n$ and $k$ with $n \geq 2k$, the first phase of Algorithm~\ref{alg:Kneser} runs in time polynomial in $n$ and either finds a monochromatic edge or produces a ground set $X_s \subseteq [n]$ and a set of colors $C_s \subseteq [n-2k+1]$ satisfying $|C_s| = |X_s|-2k+1$ and $|X_s| = O(k^4)$ (see lines~\ref{line:first_p}--\ref{line:first_p_end}).
In the latter case, the restriction of the input coloring $c$ to the vertices of the graph $K(\binom{X_s}{k})$ is not guaranteed to use only colors from $C_s$, as required by the definition of the $\KneserP$ problem.
Yet, by applying the oracle to this restriction of $c$, simulating its queries using the access to the coloring $c$ of $K(n,k)$, we either get a monochromatic edge in $K(\binom{X_s}{k})$ or find a vertex whose color does not belong to $C_s$. In both cases, a solution to the original instance can be efficiently found with high probability.
Indeed, if a monochromatic edge is returned then it forms a monochromatic edge of $K(n,k)$ as well.
Otherwise, as shown in the analysis of Algorithm~\ref{alg:Kneser}, a vertex $A \in \binom{X_s}{k}$ whose color does not belong to $C_s$ can be used to efficiently find with high probability a vertex $B \in \binom{[n]}{k}$ such that $\{A,B\}$ is a monochromatic edge (see lines~\ref{line:bb_s}--\ref{line:bb_e}).

\section{An Algorithm for the $\KneserP$ Problem Based on Schrijver Graphs}\label{sec:AlgoSchr}

In this section we present the simple deterministic algorithm for the $\KneserP$ problem given in Theorem~\ref{thm:KneserAlgSch}.
Let us first state a simple fact on the number of vertices in Schrijver graphs (see Section~\ref{sec:KneSchr}).
We include a quick proof for completeness.

\begin{fact}\label{fact:S(n,k)}
For integers $k \geq 2$ and $n \geq 2k$, the number of vertices in $S(n,k)$ is $\binom{n-k+1}{k} - \binom{n-k-1}{k-2}$.
\end{fact}

\begin{proof}
Recall that the vertex set $\binom{[n]}{k}_{\mathrm{stab}}$ of $S(n,k)$ is the collection of all $k$-subsets of $[n]$ with no two consecutive elements modulo $n$.
We first observe that the number of $k$-subsets of $[n]$ with no two consecutive elements, allowing both $1$ and $n$ to be in the subsets, is $\binom{n-k+1}{k}$.
To see this, identify the subsets of $[n]$ with their characteristic vectors in $\{0,1\}^n$, and in every such vector, interpret the zeros as balls and the ones as separations between bins. It follows that every such set corresponds to a partition of $n-k$ identical balls into $k+1$ bins, where no bin but the first and last is empty. The number of those partitions is equal to the number of partitions of $n-2k+1$ identical balls into $k+1$ bins, which is $\binom{n-k+1}{k}$.
Finally, to obtain the number of vertices of $S(n,k)$, one has to subtract the number of $k$-subsets of $[n]$ with no two consecutive elements that include both $1$ and $n$. The latter is equal to the number of $(k-2)$-subsets of $[n-4]$ with no two consecutive elements, which by the above argument equals $\binom{n-k-1}{k-2}$, so we are done.
\end{proof}

Equipped with Fact~\ref{fact:S(n,k)}, we are ready to prove Theorem~\ref{thm:KneserAlgSch}.

\begin{proof}[ of Theorem~\ref{thm:KneserAlgSch}]
Consider the algorithm that given integers $n$ and $k$ with $n \geq 2k$ and an oracle access to a coloring $c: \binom{[n]}{k} \rightarrow [n-2k+1]$ of the vertices of the Kneser graph $K(n,k)$, enumerates all vertices of the Schrijver graph $S(n,k)$, i.e., the sets of $\binom{[n]}{k}_{\mathrm{stab}}$, and queries the oracle for their colors. Then, the algorithm goes over all pairs of those vertices and returns a pair that forms a monochromatic edge. The existence of such an edge follows from Theorem~\ref{thm:KneserS}, which asserts that $\chi(S(n,k))= n-2k+2$. The running time of the algorithm is polynomial in the number of vertices of $S(n,k)$, which by Fact~\ref{fact:S(n,k)} does not exceed $\binom{n-k+1}{k} = \binom{n-k+1}{n-2k+1} \leq n^{\min(k,n-2k+1)}$.
\end{proof}

\section{The $\Agree$ Problem}\label{sec:agree}

In this section we study the $\Agree$ problem.
After presenting its formal definition, we provide efficient algorithms for families of instances of the problem and explore its connections to the $\KneserP$ problem.

\subsection{The Definition}
For a collection $M$ of $m$ items, consider $\ell$ utility functions $u_i : P(M) \rightarrow \Q^{\geq 0}$ for $i \in [\ell]$ that map every subset of $M$ to a non-negative value. The functions $u_i$ are assumed to be monotone, that is, $u_i(A) \leq u_i(B)$ whenever $A \subseteq B$.
We refer to $u_i$ as the utility function associated with  the $i$th agent.
A set $S \subseteq M$ is said to be {\em agreeable} to agent $i$ if $u_i(S) \geq u_i(M \setminus S)$, that is, the $i$th agent values $S$ at least as much as its complement.
The following theorem was proved by Manurangsi and Suksompong~\cite{ManurangsiS19} (see~\cite{GoldbergHIMS20} for an alternative proof).
\begin{theorem}[\cite{ManurangsiS19}]\label{thm:MS}
For every collection $M$ of $m$ items and for every $\ell$ agents with monotone utility functions, there exists a set $S \subseteq M$ of size
\[ |S| \leq \min \Big ( \Big \lfloor \frac{m+\ell}{2} \Big \rfloor, m \Big )\]
that is agreeable to all agents.
\end{theorem}
\noindent
The $\Agree$ problem is defined as follows (see Remark~\ref{remark:syntactic}).
\begin{definition}\label{def:Agree}
In the $\Agree$ problem, given a collection $M$ of $m$ items and $\ell$ monotone utility functions $u_i : P(M) \rightarrow \Q^{\geq 0}$ for $i \in [\ell]$ associated with $\ell$ agents, the goal is to find a set $S \subseteq M$ of size $|S| \leq \min ( \lfloor \frac{m+\ell}{2} \rfloor, m )$ that is agreeable to all agents.
In the black-box input model, the utility functions are given as an oracle access that for $i \in [\ell]$ and $A \subseteq M$ returns $u_i(A)$.
In the white-box input model, the utility functions are given by a Boolean circuit that for $i \in [\ell]$ and $A \subseteq M$ computes $u_i(A)$.
\end{definition}
\noindent
By Theorem~\ref{thm:MS}, every instance of the $\Agree$ problem has a solution.
Note that for instances with $\ell \geq m$, the collection $M$ forms a proper solution.

\subsection{Algorithms for the $\Agree$ Problem}\label{sec:Algo_Agree}

Our algorithms for the $\Agree$ problem are obtained by combining our algorithms for the $\KneserP$ problem with the relation discovered in~\cite{ManurangsiS19} between the problems.
As mentioned before, using the proof idea of Theorem~\ref{thm:MS} in~\cite{ManurangsiS19}, it can be shown that the $\Agree$ problem is efficiently reducible to the $\KneserP$ problem.
For completeness, we describe below the reduction and some of its properties.
Note that we consider here the problems in the black-box input model, but the same reduction is applicable in the white-box input model as well (see Section~\ref{sec:models}).

\paragraph{The reduction.}
Consider an instance of the $\Agree$ problem, i.e., a collection $M=[m]$ of items and $\ell$ monotone utility functions $u_i : P(M) \rightarrow \Q^{\geq 0}$ for $i \in [\ell]$.
It can be assumed that $\ell < m$, as otherwise the collection $M$ forms a proper solution, and the reduction can produce an arbitrary instance.
Put $k = \lceil\frac{m-\ell}{2}\rceil$, and define a coloring $c: \binom{[m]}{k} \rightarrow [\ell]$ of the vertices of the Kneser graph $K(m,k)$ as follows. For a $k$-subset $A \subseteq [m]$, $c(A)$ is the smallest $i \in [\ell]$ satisfying $u_i(A) > u_i(M \setminus A)$ if such an $i$ exists, and $c(A) = \ell$ otherwise.
(Note that in the latter case, $c(A)$ can be defined arbitrarily, and we define it to be $\ell$ just for concreteness.)
The definition of $k$ implies that $m-2k+1 \geq m-2 \cdot (\frac{m-\ell+1}{2})+1 = \ell$, hence the coloring $c$ uses colors from $[m-2k+1]$, as required.

\paragraph{Correctness.}
For $\ell<m$, consider a solution to the constructed $\KneserP$ instance, that is, two disjoint $k$-subsets $A,B$ of $[m]$ satisfying $c(A)=c(B)$.
Consider the complement sets $M \setminus A$ and $M \setminus B$ whose size satisfies $|M \setminus A| = |M \setminus B| = m- k =\lfloor \frac{m+\ell}{2}\rfloor$.
We claim that at least one of them is agreeable to all agents and thus forms a proper solution to the given instance of the $\Agree$ problem.
To see this, let $i \in [\ell]$ denote the color of $A$ and $B$.
Suppose in contradiction that neither of the sets $M \setminus A$ and $M \setminus B$ is agreeable to all agents, hence each of them has an agent that prefers its complement. By the definition of the coloring $c$, it follows that $u_i(A) > u_i(M \setminus A)$ and $u_i(B) > u_i(M \setminus B)$. This implies that
\[ u_i(A) > u_i(M \setminus A) \geq  u_i(B) > u_i(M \setminus B) \geq u_i(A),\]
where the second and fourth inequalities follow from the monotonicity of the function $u_i$ and the fact that $A$ and $B$ are disjoint. This implies that $u_i(A) > u_i(A)$, a contradiction.

Note that given the vertices $A$ and $B$ of a monochromatic edge in the Kneser graph $K(m,k)$ it is possible to efficiently check which of the sets $M \setminus A$ and $M \setminus B$ forms a proper solution to the $\Agree$ instance by querying the oracle for the values of $u_i(A)$, $u_i(M \setminus A)$, $u_i(B)$, $u_i(M \setminus B)$ for all $i \in [\ell]$.

\paragraph{Black-box model.}
The reduction is black-box in the following manner.
Given an oracle access to an instance of the $\Agree$ problem with parameters $\ell<m$, it is possible to simulate an oracle access to the above instance of the $\KneserP$ problem with parameters $m$ and $k = \lceil \frac{m-\ell}{2} \rceil$, such that any solution for the latter efficiently yields a solution for the former. Notice that the simulation of the oracle access to the coloring $c$ requires, for any given vertex $A \in \binom{[m]}{k}$, performing at most $2 \ell$ queries to the utility functions of the original instance to determine the values of $u_i(A)$ and $u_i(M \setminus A)$ for all $i \in [\ell]$.

\paragraph{}
The above reduction yields the following proposition.

\begin{proposition}\label{prop:reduction}
Suppose that there exists a (randomized) algorithm for the $\KneserP$ problem that given an oracle access to a coloring $c: \binom{n}{k} \rightarrow [n-2k+1]$ of the vertices of the Kneser graph $K(n,k)$ finds a monochromatic edge in running time $t(n,k)$ (with some given probability).
Then, there exists a (randomized) algorithm for the $\Agree$ problem that given an oracle access to an instance with $m$ items and $\ell$ agents ($\ell<m$), returns a proper solution in running time $t(m, \lceil \frac{m-\ell}{2} \rceil) \cdot m^{O(1)}$ (with the same probability).
\end{proposition}

Our first algorithm for the $\Agree$ problem, given in Theorem~\ref{thm:AlgoIntro1}, is obtained by combining Theorem~\ref{thm:AlgoKneserNew} and Proposition~\ref{prop:reduction}.

\begin{proof}[ of Theorem~\ref{thm:AlgoIntro1}]
By Theorem~\ref{thm:AlgoKneserNew}, there exists a randomized algorithm for the $\KneserP$ problem that given a coloring of $K(n,k)$ with $n-2k+1$ colors runs in time $t(n,k) \leq n^{O(1)} \cdot k^{O(k)}$ and finds a monochromatic edge with probability $1-2^{-\Omega(n)}$.
Combining this algorithm with Proposition~\ref{prop:reduction}, it follows that there exists a randomized algorithm for the $\Agree$ problem that given an instance with $m$ items and $\ell$ agents ($\ell<m$), returns a proper solution in time $t(m,k) \cdot m^{O(1)} \leq m^{O(1)} \cdot k^{O(k)}$ for $k = \lceil \frac{m-\ell}{2}\rceil$ with probability $1-2^{-\Omega(m)}$.
\end{proof}

As a corollary, we derive that the problem can be solved in polynomial time when the number of agents $\ell$ is not much smaller than the number of items $m$.
\begin{corollary}\label{cor:m<=l+log/loglog}
For every constant $\tilde{c} >0$, there exists a randomized algorithm that given an oracle access to an instance of the $\Agree$ problem with $m$ items and $\ell$ agents such that \[\ell \geq m - \tilde{c} \cdot \frac{\log m}{\log \log m}\]
runs in time polynomial in $m,\ell$ and returns a proper solution with probability $1-2^{-\Omega(m)}$.
\end{corollary}

\begin{proof}
Fix a constant $\tilde{c} >0$, and consider an instance of the $\Agree$ problem with $m$ items and $\ell$ agents such that $\ell \geq m - \tilde{c} \cdot \frac{\log m}{\log \log m}$. Consider the algorithm that for $\ell \geq m$, returns the collection of all items, and for $\ell < m$, calls the algorithm given in Theorem~\ref{thm:AlgoIntro1}.
Observe that in the latter case, the integer $k = \lceil \frac{m-\ell}{2}\rceil$ satisfies $k^{O(k)} = m^{O(1)}$, hence the proof is completed.
\end{proof}

Our second algorithm for the $\Agree$ problem, given in Theorem~\ref{thm:AlgoIntro2}, is obtained by combining Theorem~\ref{thm:KneserAlgSch} and Proposition~\ref{prop:reduction}.

\begin{proof}[ of Theorem~\ref{thm:AlgoIntro2}]
By Theorem~\ref{thm:KneserAlgSch}, there exists an algorithm for the $\KneserP$ problem that given a coloring of $K(n,k)$ with $n-2k+1$ colors runs in time $t(n,k)$ and finds a monochromatic edge, where $t(n,k)$ is polynomial in $\binom{n-k+1}{k} \leq n^{\min(k,n-2k+1)}$.
Combining this algorithm with Proposition~\ref{prop:reduction}, it follows that there exists an algorithm for the $\Agree$ problem that given an instance with $m$ items and $\ell$ agents ($\ell<m$), returns a proper solution in time polynomial in $t(m,k)$ for $k = \lceil \frac{m-\ell}{2}\rceil$, as required.
\end{proof}

As a corollary, we derive that the problem can be solved in polynomial time when the number of agents $\ell$ is a fixed constant.
\begin{corollary}\label{cor:constant_l}
For every constant $\ell$, there exists an algorithm that given an oracle access to an instance of the $\Agree$ problem with $m$ items and $\ell$ agents finds a proper solution in time polynomial in $m$.
\end{corollary}

\begin{proof}
Fix a constant $\ell \geq 1$, and consider an instance of the $\Agree$ problem with $m$ items and $\ell$ agents. Consider the algorithm that for $\ell \geq m$, returns the collection of all items, and for $\ell < m$, calls the algorithm given in Theorem~\ref{thm:AlgoIntro2}.
Observe that in the latter case, the integer $k = \lceil \frac{m-\ell}{2}\rceil$ satisfies $m-2k+1 \leq \ell+1$, hence the proof is completed.
\end{proof}

\subsection{Subset Queries}

In what follows, we consider a variant of the $\KneserP$ problem in the white-box input model, where the Boolean circuit that represents the coloring of the graph allows an extended type of queries, called {\em subset queries}.
Here, for a coloring $c$ of $K(n,k)$ with $n-2k+1$ colors, the input of the circuit is a pair $(i,B)$ of a color $i \in [n-2k+1]$ and a set $B \subseteq [n]$, and the circuit returns on $(i,B)$ whether $B$ contains a vertex colored $i$.
Note that this circuit enables to efficiently determine the color $c(A)$ of a given vertex $A$ by computing the output of the circuit on all the pairs $(i,A)$ with $i \in [n-2k+1]$. Moreover, by computing the output of the circuit on the pair $(c(A),[n]\setminus A)$, one can determine whether $A$ lies on a monochromatic edge. It can further be seen that if $A$ does lie on a monochromatic edge then the other endpoint of such an edge can be found efficiently.

\begin{remark}\label{remark:syntactic}
The existence of solutions to instances of problems in the complexity class $\TFNP$ should follow from a {\em syntactic} guarantee of totality. However, the Boolean circuit that represents a coloring with subset queries in the $\KneserP$ problem is not syntactically guaranteed to be consistent. This issue can be addressed by allowing in the definition of the problem a solution that demonstrates a violation of the input coloring.
In fact, such a modification is also needed in the definition of the $\Agree$ problem (see Definition~\ref{def:Agree}). There, one has to allow a solution that demonstrates a violation of the monotonicity of an input utility function.
For simplicity of presentation, we skip this issue in the proofs.
\end{remark}

We turn to prove Theorem~\ref{thm:KneserVSAgree}, which implies that the $\Agree$ problem is at least as hard as the $\KneserP$ problem with subset queries.
We consider here the problems in the white-box input model (see Section~\ref{sec:models} and Definitions~\ref{def:KneserSProblems} and~\ref{def:Agree}).

\paragraph{The reduction.}
Consider an instance of the $\KneserP$ problem, a coloring $c: \binom{[n]}{k} \rightarrow [n-2k+1]$ of the Kneser graph $K(n,k)$ for integers $n$ and $k$ satisfying $n \geq 2k$.
We define an instance of the $\Agree$ problem that consists of a collection $M = [n]$ of $n$ items and $\ell = n-2k+1$ agents. The utility function $u_i : P(M) \rightarrow \{0,1\}$ of the $i$th agent is defined for $i \in [\ell]$ as follows. For every $S \subseteq M$, $u_i(S) = 1$ if there exists a $k$-subset $A$ of $S$ such that $c(A)=i$, and $u_i(S)=0$ otherwise.
Observe that each function $u_i$ is monotone, because if a set $S$ contains a vertex of $K(n,k)$ colored $i$ so does any set that contains $S$.

\paragraph{Correctness.}
To prove the correctness of the reduction, consider a solution to the constructed instance of the $\Agree$ problem, i.e., a set $S \subseteq M$ of size
\[|S| \leq \min \Big ( \Big \lfloor \frac{n+\ell}{2} \Big \rfloor, n \Big ) = \min (n-k,n) = n-k \]
that is agreeable to all agents.
Note that its complement $M \setminus S$ has size at least $k$. Let $A$ be a $k$-subset of $M \setminus S$, and let $i$ denote the color of $A$ according to the coloring $c$.
By the definition of $u_i$, it holds that $u_i(M \setminus S)=1$, and since $S$ is agreeable to agent $i$, it follows that $u_i(S)=1$ as well. Hence, there exists a $k$-subset $B$ of $S$ such that $c(B)=i$, implying that the vertices $A$ and $B$ form a monochromatic edge in $K(n,k)$.

We claim that given the solution $S$ of the constructed $\Agree$ instance, it is possible to find in polynomial time a solution to the original instance of the $\KneserP$ problem. Indeed, given the solution $S$ one can define $A$ as an arbitrary $k$-subset of $M \setminus S$ and determine its color $i$ using the given Boolean circuit. As shown above, it follows that $u_i(S)=1$.
Then, it is possible to efficiently find a $k$-subset $B$ of $S$ with $c(B)=i$. To do so, we maintain a set initiated as $S$, and remove elements from the set as long as it still contains a vertex of color $i$ (which can be checked in polynomial time using the Boolean circuit), until we get the desired set $B$ of size $k$.
This gives us the monochromatic edge $\{A,B\}$ of the $\KneserP$ instance.

\paragraph{White-box model.}
It follows from the description of the reduction that the Boolean circuit that represents the input coloring with subset queries precisely represents the utility functions of the constructed instance of the $\Agree$ problem. In particular, the reduction can be implemented in polynomial running time.

\section*{Acknowledgements}

We are grateful to Andrey Kupavskii for helpful discussions and to the anonymous reviewers for their very useful suggestions.

\bibliographystyle{alpha}
\bibliography{mono_kneser}

\end{document}